\documentclass[pra,aps,showpacs,twocolumn,superscriptaddress]{revtex4}
\usepackage{amsmath,amsfonts,amssymb,caption,hyperref,color,epsfig,graphics,subfigure,graphicx,latexsym,mathrsfs,revsymb,theorem,url,verbatim,epstopdf}
\allowdisplaybreaks[3]

\hypersetup{colorlinks,linkcolor={blue},citecolor={blue},urlcolor={red}}

\usepackage{hyperref}

\makeatletter

\newcommand{\Rmnum}[1]{\expandafter\@slowromancap\romannumeral #1@}
\makeatother

\newtheorem{Lemma}{Lemma}

\newtheorem{Theorem}{Theorem}
\newtheorem{Corollary}{Corollary}

\def\squareforqed{\hbox{\rlap{$\sqcap$}$\sqcup$}}
\def\qed{\ifmmode\squareforqed\else{\unskip\nobreak\hfil
		\penalty50\hskip1em\null\nobreak\hfil\squareforqed
		\parfillskip=0pt\finalhyphendemerits=0\endgraf}\fi}
\def\endenv{\ifmmode\;\else{\unskip\nobreak\hfil
		\penalty50\hskip1em\null\nobreak\hfil\;
		\parfillskip=0pt\finalhyphendemerits=0\endgraf}\fi}
\newenvironment{proof}{\noindent \textbf{{Proof.~} }}{\qed}
\def\Dbar{\leavevmode\lower.6ex\hbox to 0pt
	{\hskip-.23ex\accent"16\hss}D}
\makeatletter
\def\url@leostyle{%
	\@ifundefined{selectfont}{\def\UrlFont{\sf}}{\def\UrlFont{\small\ttfamily}}}
\makeatother
\def\btb{\begin{tabular}}
	\def\etb{\end{tabular}}
\newcommand{\bra}[1]{\langle#1|}
\newcommand{\ket}[1]{|#1\rangle}

\newcommand{\tr}{\mathrm{tr}}

\begin{document}
	\title{Entanglement Polygon Inequalities for A Class of Mixed States}
	
	\author{Xian Shi}\email[]
	{shixian01@gmail.com}
\affiliation{College of Information Science and Technology, Beijing University of Chemical Technology, Beijing 100029, China}
	%\affiliation{Institute of Mathematics, Academy of Mathematics and Systems Science, Chinese Academy of Sciences, Beijing, 100190, China}
	%\author{Lin Chen}\email[]{linchen@buaa.edu.cn (corresponding author)}
	%\affiliation{School of Mathematical Sciences, Beihang University, Beijing 100191, China}
	%\affiliation{International Research Institute for Multidisciplinary Science, Beihang University, Beijing 100191, China}
	
	%\author{Yi Shen}\email[]
	%{yishen@buaa.edu.cn}
	%\affiliation{School of Mathematics and Systems Science, Beihang University, Beijing 100191, China}
	%
	%\author{Yize Sun}
	%\affiliation{School of Mathematics and Systems Science, Beihang University, Beijing 100191, China}
	
	%\author{Lijun Zhao}
	%\affiliation{School of Mathematics and Systems Science, Beihang University, Beijing 100191, China}
	
	%\author{Yumin Guo}
	%\affiliation{School of Mathematical Sciences, Capital Normal University, Beijing 100048, China}
	
	\date{\today}
	

\begin{abstract}
\indent The study on the entanglement polygon inequality of multipartite systems has attracted much attention. However, most of the results are on pure states. Here we consider the property for a class of mixed states, which are the reduced density matrices of generalized W-class states in multipartite higher dimensional systems. First we show the class of mixed states satisfies the entanglement polygon inequalities in terms of Tsallis-q entanglement, then we propose a class of tighter inequalities for mixed states in terms of Tsallis-q entanglement. At last, we get an inequality for the mixed states, which can be regarded as a relation for bipartite entanglement.
\end{abstract}

\maketitle
\section{introduction}
Entanglement is an essential feature in quantum mechanics compared with the classical \cite{horodecki2009quantum,plenio2014introduction}. It serves a key role in many information tasks, such as quantum teleportation \cite{bennett1993teleporting,bouwmeester1997experimental}, quantum dense coding \cite{harrow2004superdense}, and quantum cryptography \cite{ekert1991quantum,mirhosseini2015high}.

One of the most meaningful properties of entanglement in multipartite systems is that entanglement cannot be freely shared. Monogamy of entanglement (MoE) \cite{coffman2000distributed,zong2022monogamy} is one manifestation of the property. For a tripartite entangled state $\rho_{ABC}$, the mathematical representation of MoE in terms of an entanglement measure $E$ is that 
\begin{align}
E_{A|B}+E_{A|C}\le E_{A|BC}.\label{f1}
\end{align}
The above inequality means that the sum of entanglement in $AB$ and $AC$ is bounded by entanglement between subsystem $A$ and the remaining $BC$. This relation is valid in terms of concurrence \cite{coffman2000distributed}, the squashed entanglement measure \cite{christandl2004squashed}, and Tsallis-$q$ entanglement \cite{kim2010tsallis,luo2016general} for $n$-qubit systems. Except (\ref{f1}), other representations of MoE have been proposed \cite{gour2018monogamy,guo2020multipartite,shi2021multilinear,jin2022new,guo2023complete}. 

Recently, the other entanglement distribution property, entanglement polygon inequality (EPI), was proposed for pure states in multiqubit systems in terms of some entanglement measure $E$ \cite{qian2018entanglement},
\begin{align*}
E_{j|\overline{j}}\le \sum_{k\ne j}E_{k|\overline{k}},
\end{align*}
here $j$ denotes a subsystem $\mathcal{H}_j$, $\overline{j}$ is the remaining subsystems. Recently, EPI was proved valid for pure states in multiqudit systems in terms of some entanglement measures \cite{yang2022entanglement,shi2023entanglement}. Based on the property, a method to build the genuine entanglement measure for multipartite systems was proposed \cite{xie2021triangle}. However, there are few results on EPI for mixed states of multipartite systems in terms of some entanglement measure.

The generalized $W$ class (GW) states of multipartite higher dimensional systems own fine properties on their entanglement distribution. In 2008, Kim and Sanders showed that the GW states in multipartite higher dimensional systems satisfy the monogamy relations in terms of the squared concurrence \cite{san2008generalized}. Moreover, there the authors showed that inequality (\ref{f1}) is saturated for the GW states in terms of the squared concurrence. In recent years, there have been many results on monogamy of entanglement for the GW states in terms of convex roof extended negativity \cite{choi2015negativity}, Tsallis-$q$ entanglement \cite{shi2020monogamy}, Renyi-$\alpha$ entanglement \cite{lai2021tighter}, and the unified-$(q,s)$ entanglement \cite{li2024monogamy}.

 This article is organized as follows. In section \ref{sec2}, we present the preliminary knowledge of this article. In section \ref{sec3}, we present our main results. First, we present a class of tripartite mixed states satisfy the EPI in terms of Tsallis-$q$ entanglement, which are generated by the partial trace map over a generalized $W$-class state in terms of Tsallis-$q$ entanglement, and then we show a tighter inequality for the mixed states. At last, we present another type of inequalities for the mixed states. In section \ref{sec 4}, we end with a summary. 
 
\section{Preliminary Knowledge}\label{sec2}

In this section, we present the preliminary knowledge needed here. First, we present the definitions and properties of entanglement measures for bipartite systems, and then we recall the generalized $W$-class states in multipartite higher dimensional systems. 
\subsection{Entanglement measures}
\indent Given a bipartite pure state $\ket{\psi}_{AB}=\sum_i\sqrt{\lambda_i}\ket{ii},$
the concurrence is defined as
\begin{align}
C(\ket{\psi}_{AB})=\sqrt{2(1-\tr\rho_A^2)},\label{p1}
\end{align}
where $\rho_A=\tr_B\rho_{AB}$. When $\rho_{AB}$ is a mixed state, its concurrence is defined as
\begin{align}
C(\rho_{AB})=\min_{\{p_i,\ket{\psi_i}_{AB}\}}\sum_i p_iC(\ket{\psi_i}_{AB}),
\end{align}
where the minimum takes over all the decompositions of $\rho_{AB}=\sum_i p_i\ket{\psi_i}_{AB}\bra{\psi_i}.$ 

For a quantum state $\rho$, the Tsallis-$q$ entropy is 
\begin{align*}
T_q(\rho)=\frac{1}{q-1}[1-\tr\rho^q],
\end{align*}
for $q\ge 1.$  When $q\rightarrow 1$, $T_q(\cdot)$ converges to the von Neumann entropy $S(\cdot)$ 
\begin{align}
\lim\limits_{q\rightarrow1}T_q(\rho)=-\tr\rho\log\rho=S(\rho).\label{f2}
\end{align}

For a bipartite pure state $\ket{\psi}_{AB}$, when $q\ge 1$, its Tsallis-$q$ entanglement is 
\begin{align*}
\mathbf{T}_q(\ket{\psi}_{AB})=T_q(\rho_A),
\end{align*}
where $\rho_A=\tr_B\ket{\psi}_{AB}\bra{\psi}$. For a mixed state $\rho_{AB}$, its Tsallis-$q$ entanglement is defined as 
\begin{align*}
\mathbf{T}_q(\rho_{AB})=\min\sum_i p_i \mathbf{T}_q(\ket{\psi_i}_{AB}),
\end{align*}
where the minimum takes over all possible pure state decompositions of $\rho_{AB}=\sum_i p_i\ket{\psi_i}_{AB}\bra{\psi_i}.$

By combing (\ref{f2}), we have 

\begin{align*}
\lim\limits_{q\rightarrow 1} \mathbf{T}_q(\rho_{AB})=E(\rho_{AB}),
\end{align*}
where $E(\rho_{AB})$ is entanglement of formation for the mixed state $\rho_{AB}$ \cite{Bennett1996Concentrating,Bennett1996mixed}.

Next assume $\ket{\phi}$ is a pure state in $\mathcal{H}_2\otimes\mathcal{H}_d$ with its Schmidt decomposition $\ket{\phi}_{AB}=\sqrt{\lambda}\ket{00}+\sqrt{1-\lambda}\ket{11},$ its Tsallis-$q$ entanglement is 
\begin{align*}
\mathbf{T}_q(\ket{\phi}_{AB})=\frac{1}{(1-q)}[\lambda^q+(1-\lambda)^q-1].
\end{align*}
Combing the definition of concurrence for pure states, we have
\begin{align*}
\mathbf{T}_q(\ket{\phi}_{AB})=f_{q}(C^2(\ket{\phi_{AB}})),
\end{align*}
where $f_{q}(x)$ is defined in $x\in[0,1]$,
\begin{align*}
f_{q}(x)=\frac{((1+\sqrt{1-x})^q+(1-\sqrt{1-x})^q)-2^{q}}{(1-q)2^{q}}.
\end{align*}

\indent Next we recall several lemmas on the properties of $f_q(x)$.
\begin{Lemma}\label{l1}\cite{kim2010tsallis}
The function $f_q(x)$ is a monotonously increasing and concave function when $q\in [\frac{5-\sqrt{13}}{2},2]\cup[3,\frac{5+\sqrt{13}}{2}]$.
\end{Lemma}

\begin{Lemma}\cite{luo2016general}
	The function $f_q(x^2)$ is a monotonically increasing function of the variable $x$ for $q\in (0,\infty)$ and $x\in(0,1)$, it is a convex function of $x$ when $q\in[\frac{5-\sqrt{13}}{2},\frac{5+\sqrt{13}}{2}].$
\end{Lemma}

\subsection{Generalized W-class State in Multipartite Higher Dimensional Systems}
 Here we first recall the generalized W-class (GW) state $\ket{W_n^d}_{A_1A_2\cdots A_n}$ in multi-qudit systems \cite{san2008generalized},
\begin{align}
&\ket{W_n^d}_{A_1A_2\cdots A_n}\nonumber\\
=&\sum_{j=1}^{d-1}(a_{1j}\ket{j0\cdots0}+a_{2j}\ket{0j\cdots0}+\cdots+a_{nj}\ket{00\cdots0j}),\label{f3}
\end{align}
here $\{\ket{j}_{A_i}\}_{j=0}^{d-1}$ is an orthonormal basis of qudit subsystems $A_i$ with $i=1,2,\cdots,n$ and $\sum_{i=1}^n\sum_{j=1}^{d-1}|a_{ij}|^2=1$.

\begin{Lemma}\label{l4}\cite{san2008generalized}
	For any n-qudit generalized W-class state $\ket{\psi}_{AB_{1}\cdots B_{n-1}}$ in ($\ref{f3}$) and a partition $P=\{P_1,\cdots,P_m\}$ for the subset $S=\{A, B_{1}, \cdots, B_{l-1}\},$ $m\le l\le n,$
	\begin{align}
	C^2_{P_s|\overline{P_s}}=\sum_{k\ne s}C^2_{P_s|P_k}=\sum_{k\ne s}(C^a_{P_s|P_k})^2,\label{c1}
	\end{align}
	for all $k\ne s$, here we denote $\{P_1,\cdots, \overline{P_s},\cdots P_m\}=\{P_1,\cdots, {P_s},\cdots, P_m\}-\{P_s\}$  
\end{Lemma}

Based on Lemma \ref{l1} and the relations between $\mathbf{T}_{q}(\cdot)$ and $C(\cdot)$, the authors in \cite{shi2020monogamy} showed that 
\begin{Lemma}\label{l6}
	Assume $\rho_{A_{j_1}A_{j_2}\cdots A_{j_m}}$ is a reduced density matrix of an $n$-qudit GW state in (\ref{f3}), here $m\le n$, then we have
	\begin{align*}
\mathbf{T}_q(\rho_{A_{j_1}|A_{j_2}\cdots A_{j_m}})=f_{q}(C^2(\rho_{A_{j_1}|A_{j_2}\cdots A_{j_m}})),
	\end{align*}
	for $q\in[\frac{5-\sqrt{13}}{2},\frac{5+\sqrt{13}}{2}]$.
\end{Lemma}

\section{Main Results}\label{sec3}
This section mainly considers classes of entanglement polygon inequalities in terms of the Tsallis-$q$ entanglement for a class of mixed states, which are the reduced density matrices of a GW state. First, we show that the EPI is satisfied by the class of mixed states. Then a tighter EPI is presented for the class of mixed states. At last, we present a generalized EPI for the mixed states, this inequality can be regarded as a relationship of a bipartite entanglement.

\begin{Theorem}\label{t1}
	Assume $\rho_{A_{j_1}A_{j_2}A_{j_3}}$ is a tripartite reduced density matrix of an $n$-qudit GW state in $\mathcal{H}_{A_1}\otimes\mathcal{H}_{A_2}\otimes\cdots\mathcal{H}_{A_n}$, then
	\begin{align*}
\mathbf{T}_q(\rho_{A_{j_1}|A_{j_2}A_{j_3}})\le \mathbf{T}_q(\rho_{A_{j_2}|A_{j_1}A_{j_3}})+\mathbf{T}(\rho_{A_{j_3}|A_{j_1}A_{j_2}}),
	\end{align*}
	when $q\in[\frac{5-\sqrt{13}}{2},2]\cup[3,\frac{5+\sqrt{13}}{2}].$
\end{Theorem}
\begin{proof}
	As $\rho_{A_{j_1}A_{j_2}A_{j_3}}$ is a tripartite reduced density matrix of a GW state, then 
	\begin{align*}
	&\mathbf{T}_{q}(\rho_{A_{j_2}|A_{j_1}A_{j_3}})+\mathbf{T}_{q}(\rho_{A_{j_3}|A_{j_1}A_{j_2}})\nonumber\\
	=&f_{q}(C^2_{A_{j_2}|A_{j_1}A_{j_3}})+f_{q}(C^2_{A_{j_3}|A_{j_2}A_{j_1}})\\
	=&f_q(C^2_{A_{j_2}A_{j_1}}+C^2_{A_{j_2}A_{j_3}})+f_q(C^2_{A_{j_3}A_{j_2}}+C^2_{A_{j_3}A_{j_1}})\\
	\ge&f_q(C^2_{A_{j_1}A_{j_2}}+C^2_{A_{j_1}A_{j_3}})\\
	=&f_q(C^2_{A_{j_1}|A_{j_2}A_{j_3}})\\
	=&\mathbf{T}_q(\rho_{A_{j_1}|A_{j_2}A_{j_3}}),
	\end{align*} 
	here the first equality is due to the Lemma \ref{l6}, the second equality is due to Lemma \ref{l4}, the first inequality is due to the Lemma \ref{l1}, and the last equality is due to Lemma \ref{l4}.
\end{proof}

Due to the proof of Theorem \ref{t1}, we can prove the following similarly, 
\begin{align*}
\mathbf{T}_q(\rho_{A_{j_2}|A_{j_1}A_{j_3}})\le \mathbf{T}_q(\rho_{A_{j_1}|A_{j_2}A_{j_3}})+\mathbf{T}(\rho_{A_{j_3}|A_{j_1}A_{j_2}}),
\end{align*}
that is,
\begin{Corollary}\label{c3}
		Assume $\rho_{A_{j_1}A_{j_2}A_{j_3}}$ is a tripartite reduced density matrix of an $n$-qudit GW state in $\mathcal{H}_{A_1}\otimes\mathcal{H}_{A_2}\otimes\cdots\mathcal{H}_{A_n}$, then
	\begin{align*}
|\mathbf{T}_q(\rho_{A_{j_2}|A_{j_1}A_{j_3}})-\mathbf{T}_q(\rho_{A_{j_3}|A_{j_1}A_{j_2}})|\le&\\	\mathbf{T}_q(\rho_{A_{j_1}|A_{j_2}A_{j_3}})\le \mathbf{T}_q(\rho_{A_{j_2}|A_{j_1}A_{j_3}})+&\mathbf{T}(\rho_{A_{j_3}|A_{j_1}A_{j_2}}),
	\end{align*}
	when $q\in[\frac{5-\sqrt{13}}{2},2]\cup[3,\frac{5+\sqrt{13}}{2}].$
\end{Corollary}

\begin{figure}
	\centering
	\includegraphics[scale=0.4]{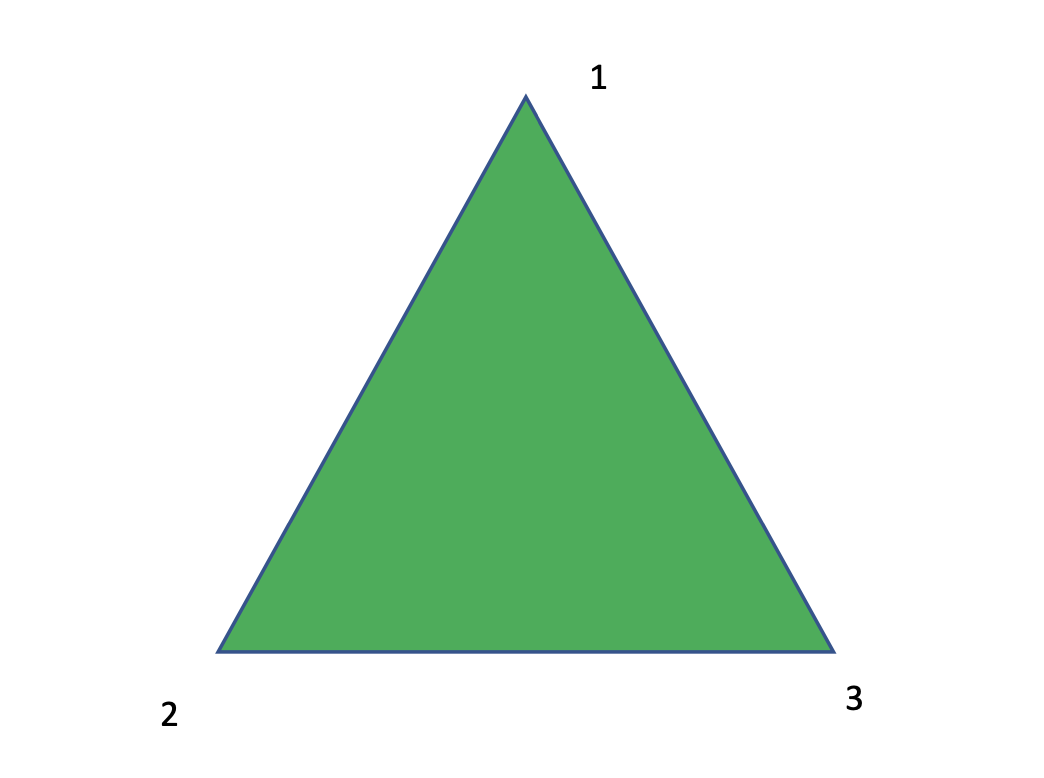}
	
	\captionsetup{justification=raggedright, singlelinecheck=false}
	
	\caption{Characterization of Corollary \ref{c3}. 
		In this figure, the length of segment $i-j$ indicates the value of $\mathbf{T}_{q}(\rho_{A_k|A_iA_j})$, here $k$ is another party different from parties $i$ and $j$.}\label{fig1}
\end{figure}

In Fig. \ref{fig1}, we plot a triangle to present the meaning of Corollary \ref{c3}, the length of segment $i-j$ denotes the entanglement of $\mathbf{T}_{q}(\rho_{A_k|A_iA_j})$, here $k$ is another party.

\begin{Corollary}
	Assume $\rho_{A_{j_1}A_{j_2}\cdots A_{j_m}}$ is a multipartite reduced density matrix of an $n$-qudit $GW$ state in $\mathcal{H}_{A_1}\otimes\mathcal{H}_{A_2}\otimes\cdots\otimes\mathcal{H}_{A_n}$, and $\{P_1,P_2,\cdots,P_h\}$ is a partition of the set $S=\{A_{j_1},A_{j_2},\cdots,A_{j_m}\}$, that is, $\cup_{g=1}^hP_g=S$,
	\begin{align*}
	\mathbf{T}_q(\rho_{P_{j_k}|\overline{P_{j_k}}})
	\le \sum_{l\ne k}\mathbf{T}_q(\rho_{P_{j_l}|\overline{P_{j_l}}}),
	\end{align*}
		here $\overline{P_{j_i}}=P_{j_0}P_{j_1}\cdots P_{j_{i-1}}P_{j_{i+1}}\cdots P_{j_m}$, $i=0,1,\cdots,m$, and $m\le n$.
\end{Corollary}

\indent Assume $j\in \mathbb{N}^{+}$, then $j$ can be written as
\begin{align}
j=\sum_{i=0}^{n-1}j_i 2^i,\label{tqa1}
\end{align}
here we assume $\log_2 j\le n,j_i\in \{0,1\}.$ According to the equality (\ref{tqa1}), we have the following bijection:
\begin{align}
j\rightarrow &\vec{j}\nonumber\\
j\rightarrow &(j_0,j_1,\cdots,j_{n-1}),\nonumber
\end{align} 
and we denote its Hamming distance $w_H(\vec{j})$ as the number 1 of the set $\{j_0,j_1,\cdots,j_{n-1}\}$. Next we present a tighter inequality of the GW states in terms of the Tsallis-$q$ entanglement.
\begin{Theorem}\label{tqa5}
	Let $\beta\in [0,1],$ assume $\rho_{A_{i_1}A_{i_2}\cdots A_{i_m}}$ is a reduced density matrix of a GW state $\ket{\psi}_{A_0A_1\cdots A_{n-1}}$. Let $P=\{P_{j_0},P_{j_1},\cdots,P_{j_{m-1}},P_{j_m}\}$ be a partition of the set $\{A_{i_1},A_{i_2},\cdots, A_{i_m}\}$, 	when $q\in[\frac{5-\sqrt{13}}{2},2]\cup[3,\frac{5+\sqrt{13}}{2}],$ there exists an appropriate order of $P_{j_0}, P_{j_1},\cdots,P_{j_{m-1}},P_{j_m}$ such that
	\begin{align}
	[\mathbf{T}_q(\rho_{P_{j_0}|\overline{P_{j_0}}})]^{\beta}\le \sum_{i=1}^{m}(2^{\beta}-1)^{w_H(\vec{j_i})}[\mathbf{T}_q(\rho_{P_{j_i}|\overline{P_{j_i}}})]^{\beta},
	\end{align}
	here $\overline{P_{j_i}}=P_{j_0}P_{j_1}\cdots P_{j_{i-1}}P_{j_{i+1}}\cdots P_{j_m}$, $i=0,1,\cdots,m$.
\end{Theorem}

Here we place the proof of Theorem \ref{tqa5} in Sec. \ref{app}.

At last, we present the other class of entanglement polygon inequalities for bipartite entanglement. Assume $\ket{\phi}_{\mathbf{A}\mathbf{B}C}$ is a GW state in $\mathcal{H}_{\mathbf{A}}\otimes\mathcal{H}_{\mathbf{B}}\otimes\mathcal{H}_C,$ $\rho_{\mathbf{A}\mathbf{B}}=\tr_C\ket{\phi}_{\mathbf{A}\mathbf{B}C}\bra{\phi}$ is a mixed state of the Hilbert space $\mathcal{H}_{\mathbf{A}\mathbf{B}}$. Here we denote $\mathbf{A}=A_1A_2\cdots A_m$ and $\mathbf{B}=B_1B_2\cdots B_m.$

\begin{Theorem}
Assume $\ket{\phi}_{\mathbf{A}\mathbf{B}C}$ is a GW state in $\mathcal{H}_{\mathbf{A}\mathbf{B}C},$ $\rho_{\mathbf{A}\mathbf{B}}=\tr_C\ket{\phi}_{\mathbf{A}\mathbf{B}C}\bra{\phi}$, here $\mathbf{A}=A_1A_2\cdots A_m$ and $\mathbf{B}=B_1B_2\cdots B_m,$ then 
\begin{align}
\mathbf{T}_q^{\mathbf{A}|\mathbf{B}}\le \sum_{i,j=1}^m\mathbf{T}_q^{A_i|B_j},
\end{align} 
	when $q\in[\frac{5-\sqrt{13}}{2},2]\cup[3,\frac{5+\sqrt{13}}{2}].$
\end{Theorem}
\begin{proof}
As $\rho_{\mathbf{A}\mathbf{B}}$ is reduced density matrix of a GW state, then
\begin{align*}
\mathbf{T}_q(\rho_{\mathbf{A}\mathbf{B}})=&f_q(C^2_{\mathbf{A}|\mathbf{B}})\\
=&f_q(\sum_{i,j=1}^m C_{A_iB_j}^2)\\
\le&\sum_{i,j=1}^m f_q(C^2_{A_iB_j})\\
=&\sum_{i,j=1}^m T_q(\rho_{A_iB_j}).
\end{align*}
\end{proof}
\section{Summary}\label{sec 4}
In this manuscript, we have mainly considered the distribution property of multipartite entanglement for a class of mixed states in terms of the Tsallis-$q$ entanglement. First, we have proved the EPI in terms of Tsallis-$q$ entanglement for the reduced density matrices of the GW states. With Hamming distance, we have presented a tighter EPI in terms of Tsallis-$q$ entanglement for the mixed states. At last, we have shown inequalities for the mixed states in terms of Tsallis-$q$ entanglement. We hope our results can provide a reference for future work on the study of multiparty quantum entanglement.
\section{Acknowledgement}
This work was supported by the National Natural Science Foundation of China (Grant No. 12301580).
\bibliographystyle{IEEEtran}
\bibliography{ref}

\section{Appendix}\label{app}
\begin{Lemma}\label{l7}
	Let $\beta\in [0,1],x\in (0,1],$ then we have 
	\begin{align}
	(1+x)^{\beta}\le 1+(2^{\beta}-1)x^{\beta}.
	\end{align}
\end{Lemma}
\begin{proof}
	Let $t=\frac{1}{x},$ then the lemma is equivalent to get the maximum of $f(t)$ when t$\in[1,\infty),$
	\begin{align}
	f(t)=(1+t)^{\beta}-t^{\beta}.
	\end{align}
	As $t\in [1,\infty),$ and $f^{'}(t)\le 0,$ that is, when $t=1,$ $f(t)$ get the maximum $2^t-1.$ At last, When we replace t with $\frac{1}{x},$ we finish the proof. 
\end{proof}
 \textbf{Theorem} \ref{tqa5}: \emph{
	Let $\beta\in [0,1],$ assume $\rho_{A_{i_1}A_{i_2}\cdots A_{i_m}}$ is a reduced density matrix of a GW state $\ket{\psi}_{A_0A_1\cdots A_{n-1}}$, let $P=\{P_{j_0},P_{j_1},\cdots,P_{j_{m-1}},P_{j_m}\}$ be a partition of the set $\{A_{i_1},A_{i_2},\cdots, A_{i_m}\}$, then there exists an appropriate order of $P_{j_0}, P_{j_1},\cdots,P_{j_{m-1}},P_{j_m}$ such that
\begin{align}
[\mathbf{T}_q(\rho_{P_{j_0}|\overline{P_{j_0}}})]^{\beta}\le \sum_{i=1}^{m}(2^{\beta}-1)^{w_H(\vec{j_i})}[\mathbf{T}_q(\rho_{P_{j_i}|\overline{P_{j_i}}})]^{\beta},
\end{align}
here $\overline{P_{j_i}}=P_{j_0}P_{j_1}\cdots P_{j_{i-1}}P_{j_{i+1}}\cdots P_{j_m}$, $i=0,1,\cdots,m$.}

\begin{proof}
In the process of the proof, we always can order the partite $P_{j_0}, P_{j_1}$ $\cdots$ $P_{j_{m-1}}$ such that 
\begin{align}
\mathbf{T}_q(\rho_{P_{j_i}|\overline{P_{j_i}}})\ge \mathbf{T}_q(\rho_{P_{j_{i+1}}|\overline{P_{j_{i+1}}}}),i=0,1,\cdots,m-1.\label{tqa2}
\end{align}
Through computation, we have that all the reduced density matrices of a GW state is can be written as 

\begin{align*}
\rho_{P_{j_0}\cdots P_{j_{m-1}}}=\gamma_{P_{j_0}\cdots P_{j_{k-1}}}\otimes\ket{0_{m-k}}\bra{0_{m-k}}
\end{align*}

\indent Here we will use the mathematical induction method to prove the theorem. When $\rho_{P_{j_0}P_{j_1}P_{j_2}}$ is a tripartite reduced density matrix of a GW state  $\ket{\psi}_{AA_1\cdots A_{n-1}}$,
\begin{align}
& [\mathbf{T}_q(\rho_{P_{j_0}|P_{j_1}P_{j_2}})]^{\beta}\nonumber\\
\le &[\mathbf{T}_q(\rho_{P_{j_1}|P_{j_0}P_{j_2}})+\mathbf{T}_q(\rho_{P_{j_2}|P_{j_0}P_{j_1}})]^{\beta}\nonumber\\
= &[\mathbf{T}_q(\rho_{P_{j_1}|P_{j_0}P_{j_2}})]^{\beta}\left[1+[\frac{\mathbf{T}_q(\rho_{P_{j_2}|P_{j_0}P_{j_1}})}{\mathbf{T}_q(\rho_{P_{j_1}|P_{j_0}P_{j_2}})}]\right]^{\beta}\nonumber\\
\le &(\mathbf{T}_q(\rho_{P_{j_1}|P_{j_0}P_{j_2}}))^{\beta}+(2^{\beta}-1)(\mathbf{T}_q(\rho_{P_{j_2}|P_{j_0}P_{j_1}}))^{\beta},\label{tqa3}
\end{align}
here the first inequality is due to the Theorem \ref{t1}, and when $a>c>0,b>0,$ $a^b>c^b,$ and the second inequality is due to the Lemma \ref{l7}.\\
\indent Next assume $m<2^{n},$ the theorem is correct. Then when $m=2^n$, from the inequality (\ref{tqa3}), we have   

\begin{align}
&[\mathbf{T}_q(\rho_{P_{j_0}|\overline{P_{j_0}}})]^{\beta}\nonumber\\
\le&[\sum_{i=1}^{m/2-1}\mathbf{T}_q(\rho_{P_{j_i}|\overline{P_{j_i}}})]^{\beta}[1+\frac{\sum_{i=\frac{m}{2}}^{m-1}\mathbf{T}_q(\rho_{P_{j_i}|\overline{P_{j_i}}})}{\sum_{i=0}^{\frac{m}{2}-1}\mathbf{T}_q(\rho_{P_{j_i}|\overline{P_{j_i}}})}]^{\beta}\nonumber\\
\le &[\sum_{i=1}^{\frac{m}{2}-1}\mathbf{T}_q(\rho_{P_{j_i}|\overline{P_{j_i}}})]^{\beta}+(2^{\beta}-1)[\sum_{i=\frac{m}{2}}^{m-1}\mathbf{T}_q(\rho_{P_{j_i}|\overline{P_{j_i}}})]^{\beta}\nonumber\\
\le& \sum_{i=1}^{\frac{m}{2}-1}(2^{\beta}-1)^{w_H(\vec{j_i})}(\mathbf{T}_q(\rho_{P_{j_i}|\overline{P_{j_i}}}))^{\beta}\nonumber\\+&\sum_{i=\frac{m}{2}}^{m-1}(2^{\beta}-1)\times(2^{\beta}-1)^{w_H(\vec{j_i})-1}[\mathbf{T}_q(\rho_{P_{j_i}|\overline{P_{j_i}}})]^{\beta}\nonumber\\
\le& \sum_{i=1}^{m-1}(2^{\beta}-1)^{w_H(\vec{j_i})}[\mathbf{T}_q(\rho_{P_{j_i}|\overline{P_{j_i}}})]^{\beta} .\label{tqa4}
\end{align}

\indent When $m$ is an arbitrary number, we always can choose an $n\in \mathbb{N}^{+}$ such that $2^{n-1}\le m\le 2^n.$ Let 
\begin{align*}
\gamma_{P_{j_0}\cdots P_{j_{2^n-1}}}=\rho_{P_{j_0}P_{j_1}\cdots P_{j_m-1}}\otimes \delta_{P_{j_m}P_{j_{m+1}}\cdots P_{j_{2^n-1}}},
\end{align*}
here $$\delta_{P_{j_m}P_{j_{m+1}}\cdots P_{j_{2^n-1}}}=\ket{00\cdots 0}_{P_{j_m}P_{j_{m+1}}\cdots P_{j_{2^n-1}}}\bra{00\cdots 0}.$$	Then due to the inequality (\ref{tqa4}), we have 
\begin{align}
[\sum_{i=0}^{2^n-1}\mathbf{T}_q(\gamma_{P_{j_i}|\overline{P_{j_i}}})]^{\beta}\le \sum_{i=0}^{2^n-1}(2^{\beta}-1)^{w_H(\vec{j_i})}[\mathbf{T}_q(\gamma_{P_{j_i}|\overline{P_{j_i}}})]^{\beta}.
\end{align}

\indent From the definition of the state $\gamma_{P_{j_0}\cdots P_{j_{2^n-1}}}$, we have 
\begin{align}
\mathbf{T}_q(\gamma_{P_{j_i}|\overline{P_{j_i}}})=& \mathbf{T}_q(\rho_{P_{j_i}|\overline{P_{j_i}}}),\hspace{5mm} i=1,\cdots,m-1,\label{f4}\\
\mathbf{T}_q(\gamma_{P_{j_i}|\overline{P_{j_i}}})=& 0,\hspace{5mm} i=m,m+1,\cdots,2^n,\label{f5}
\end{align}
hence,
\begin{align}
&[\mathbf{T}_q(\gamma_{P_{j_i}|\overline{P_{j_i}}})]^{\beta}\nonumber\\
=& [\mathbf{T}_q(\rho_{P_{j_i}|\overline{P_{j_i}}})]^{\beta}\nonumber\\
\le& \sum_{i=0}^{2^n-1}(2^{\beta}-1)^{w_H(\vec{j_i})}[\mathbf{T}_q(\gamma_{P_{j_i}|\overline{P_{j_i}}})]^{\beta}\nonumber\\
=&\sum_{i=0}^{m-1}(2^{\beta}-1)^{w_H(\vec{j_i})}[\mathbf{T}_q(\rho_{P_{j_i}|\overline{P_{j_i}}})]^{\beta},
\end{align}
here the first equality is due to (\ref{f4}), the second equality is due to (\ref{f5})
\end{proof}

\end{document}